\documentclass[12pt,a4paper]{article}
\usepackage[english]{babel}
\usepackage[utf8]{inputenc}
\usepackage[T1]{fontenc}
\usepackage[table,xcdraw]{xcolor}
\usepackage{float}

\usepackage{graphicx}
\usepackage{amsthm}
\begin{document}
\title{Complexity of Tetris variant}
\author{Oscar Temprano\\
\small{\texttt{oscartemp@hotmail.es}}}
\date{}
\maketitle

\newtheorem{teorema}{Theorem}
\newtheorem{lema}{Lemma}

\section{Introduction}

The game complexity of Tetris was first studied by Erik demaine, Susan Hohenberger and David Liben-Nowell in “Tetris is hard" \cite{Te}.\\
Their proof was later simplified in “Tetris is hard, made easy” \cite{Te2}.\\
The authors of the last two papers published a joint paper \cite{Te3},\\
In “Tetris is hard” they left open the question of the complexity of a variant of the game.\\
In this variant, we can move and rotate a piece all the time we want in the first row, and the piece will not fall unless we move it down.\\
Once we move the piece down we will not be able to move or rotate the piece anymore.\\
This way to play tetris was introduced by John Brzustowsky \cite{TeWin}.\\
In this paper we are going to prove that the problem of clearing the board in this variant of the game is NP-complete, solving a problem that has been open during thirteen years.\\
To accomplish this task we are going to reduce the 3-partition problem to the problem of clearing a particular board on that variant of tetris.\\

\section{Rules of Tetris} 

The game tetris consists of a game board and a set of pieces (tetriminos) that fall down to the game board.\\
Every tetrimino is composed of four blocks.\\
When a piece falls, it will fall on top of another piece or in the bottom of the game board.
The objective of the player is to survive as much time as possible by clearing the maximum amount of lines in the process.\\
A line is completed when is completely filled by blocks, when that happens, that line disappears.
The game ends when blocks reach the top of the board.\\
There are seven different tetriminos:\\
Left gun (LG), right gun (RG), I(I), square (Sq), left snake (LS) , right snake (RS) and T (T).\\
To see a better description of the rules of Tetris and the pieces we use, check \cite{Te}.

\section{Reduction}

Now, we are going to show the details of the reduction.

We are going to describe the 3-PARTITION problem and
we are going to see how is the initial board, and how to clear the board, if the instance of three partition is satisfiable
\\
\\
\subsection{The 3-PARTITION problem}

The 3-Partition problem can be defined as follows:

Given a multiset $A$ of positive numbers $a_1$,...,$a_3s$ and a positive value $B$ such that $B/4$ < $a_1$ < $B/2$ for all 1 $\leq i \leq 3s$ such that $\sum_{i=1}^{3s}a_{i} = sB$.
Can $A$ be divided into $s$ disjoint subsets $A_{1}$,...,$A_{s}$ such that $\sum_{a_{i}\in A_{j}} = B$ for all 1 $\leq j \leq s$

\begin{teorema}
\label{partition}
(Garey and johnson \cite{GJ} ) 3-PARTITION is np-complete in the strong sense.
\end{teorema}

\subsection{The initial board}

The width of the board is $4 \times s + 3$ (where "s" is the number of subsets in which we have to place the numbers). 
The height of the board equals $16 + T \times 4 + 5$, (where T is the number that each subset should add to), we add the number five at the end so that the player has five extra lines to decide the movements of the pieces.

Now, we are going to see a picture of the board at the start of the game, this board corresponds to the multiset ($\lbrace 4,3,2,1,1,1,2,2,2 \rbrace$).

\begin{figure}[H]
  \centering
    \includegraphics{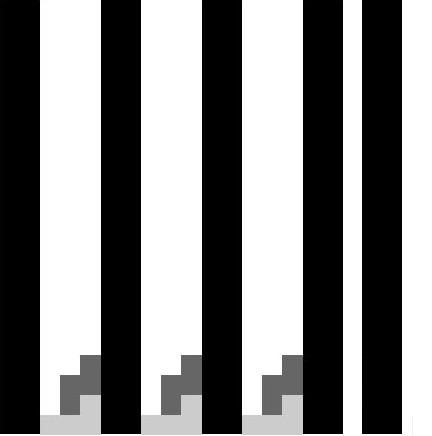}
  \caption{Example of an initial board}
  \label{fig:board}
\end{figure}

The board consists of three "buckets". Each "bucket" represents a subset in which we have to place the numbers.\\
As there are nine numbers, we have three buckets, one for each respective subset.
Each bucket has a separator at the left, and has a width of three, for a total width of four blocks per bucket.
Each bucket has an height of $16 + T \times 4$

In each bucket there are two planted pieces, a right gun, and a left snake on top of it, as can be seen in th figure.\\
Finally, at the  right of the buckets, there is a fill space with a width of one block, and two separators ,one at the left and other at the right, with a width of one or more each.

To see how to construct the board check \cite{TeCon}

Now, we are going to show how to clear the board.

First we receive a "left gun" piece. With this piece we have to choose in which bucket we are going to put the pieces that belong to the first number.

To do that we must put the piece in top  of the bucket to "open" it, as seen in the following figure:

\begin{figure}[H]
  \centering
    \includegraphics{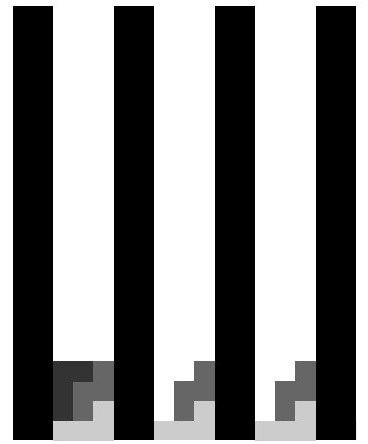}
  \caption{How to open a bucket}
  \label{fig:open}
\end{figure}

Next, we receive the pieces for every number that belongs to the multiset. Each digit of these pieces is formed by two "T" pieces and a right gun.

We place the pieces of a particular digit like this:

\begin{figure}[H]
  \centering
    \includegraphics{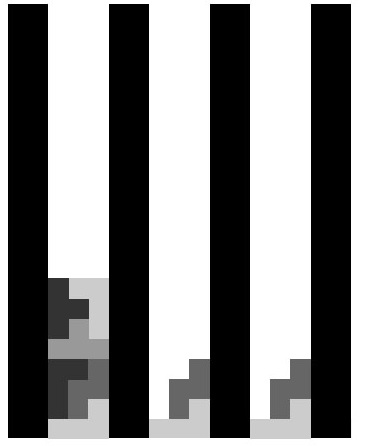}
  \caption{First digit}
  \label{fig:first_digit}
\end{figure}

\newpage

In our example, the first digit is four, so we have to place four digits:

\begin{figure}[H]
  \centering
    \includegraphics{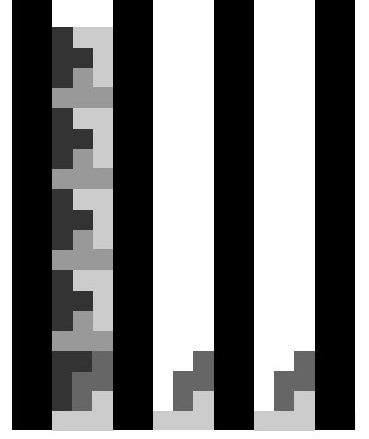}
  \caption{First number}
  \label{fig:first_numer}
\end{figure}

\newpage

Finally, the only remaining thing to do is to close the bucket, to do that, we receive a right gun and a left snake, we must put these pieces on top of the opened bucket:

\begin{figure}[H]
  \centering
    \includegraphics{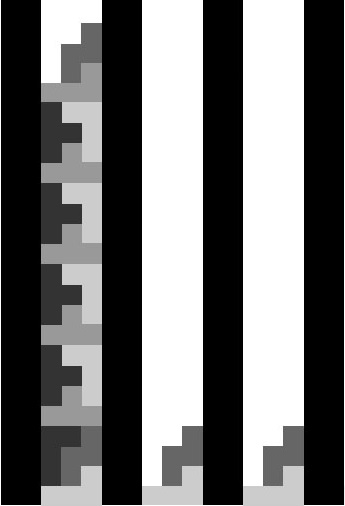}
  \caption{Bucket closed}
  \label{fig:close_bucket}
\end{figure}

We repeat this sequence of movements for every number that belongs to the multiset.

So, for each number we reecive:

- A left gun to open a bucket.

- For every digit of a number: two "T" pieces and a right gun.

- A right gun and a left snake to close the bucket.

At the end of the sequence of numbers, we will have gaps in top of the buckets:

\begin{figure}[H]
  \centering
    \includegraphics{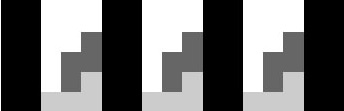}
  	\label{fig:end_numbers}
\end{figure}

To clear them we receive $s$ left guns to clear the gaps 

And finally to clear the board we receive $4 + T$  pieces of type "I" to clear the fill region and the board ("T" is the quantity that the numbers of each subset should add to).

To clear the fill region, we simply put the "I" pieces in vertical position, each one on top of another.

\section{Correctness of the reduction}

In the previous section, we have seen how to reduce a positive instance of 3-PARTITION to an instance of Tetris where the board can be cleared.

Now we are going to show that every negative instance of 3-PARTITION reduces to a game of Tetris where the board cannot be cleared. We will do this through a series of lemmas.

\begin{lema}
\label{one}
If we try to clear pieces of the fill space with other piece other than the "I" piece we will not be able to clear the board
\end{lema}

\begin{proof}
Any other piece different than the I piece cannot go into the fill space without blocking it, because the I piece is the only piece that has a width of one block.
If we try to put another piece into it we will block the fill space and we will have to clear a line extra to "unlock" it.\\
As we have just enough pieces to clear the board, we will not be able to clear the board
\end{proof}

\begin{lema}
\label{two}
If we try to use a "I" piece to cover gaps of a bucket we will not we able to clear the board
\end{lema}

\begin{proof}
Notice that we receive just enough "I" pieces to clear the fill part, if we try to use a "I" piece to cover gaps in other bucket, we will lack pieces to cover the fill space.
By lemma one we showed that if we try to cover the fill space with other pieces different than the "I" piece we will not be able to clear the board
\end{proof}

\begin{lema}
\label{three}
If there is ever placed a piece above the $16 + 4 \times T$ bottom lines we will not be able to clear the board
\end{lema}

\begin{proof}
Notice that we have just enough pieces to clear the board, this is independent of the 3-PARTITION instance being positive or negative.

Notice too, that in Tetris once a piece has been fixed into a stable place, it will only fall down if the pieces behind it are cleared, it will not fall otherwise, even if there are gaps that are not filled under it

So, if we have a piece above of the height of the buckets, we will not be able to clear the board, because we will have to clear more lines than $16 + 4 \times T$ and as we have just enough pieces to clear the buckets, we will not be able to clear the board
\end{proof}

\begin{lema}
\label{four}
If some gaps that cannot be covered by other piece different than the "I" piece are created, we will not be able to clear the board
\end{lema}

\begin{proof}
As seen in lemma two, if we try to use an "I" piece to clear gaps of the buckets, we will not be able to clear the board.

As we have just enough pieces to fill the buckets, if we create gaps that cannot be filled, some buckets will have pieces above the top of the buckets.
And by lemma three, we will not be able to clear the board 
\end{proof}

\begin{lema}
\label{five}
If we put a "right gun" piece inside a closed bucket we will not be able to clear the board
\end{lema}

\begin{proof}

We will show this with a figure:

\begin{figure}[H]
  \centering
    \includegraphics[scale=0.5]{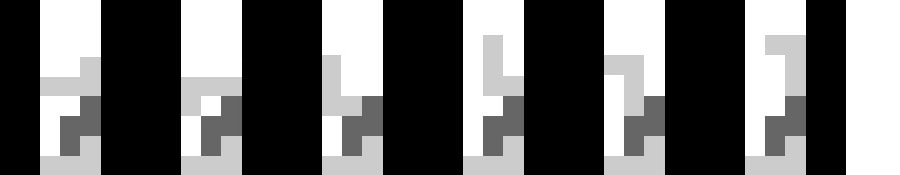}
  \label{fig:RG_closed}
\end{figure}

This figure shows the six different ways a right gun can be placed inside a closed bucket

As can be seen in the figure, if we put a right gun in top of a closed bucket we will create gaps that cannot be filled. And as it was shown in lemma four, we will not be able to clear a board with gaps that cannot be filled

\end{proof}

\begin{lema}
\label{six}
If we put a "T" piece inside a closed bucket we will not be able to clear the board
\end{lema}

\begin{proof}

We will show this with a figure:

\begin{figure}[H]
  \centering
    \includegraphics[scale=0.5]{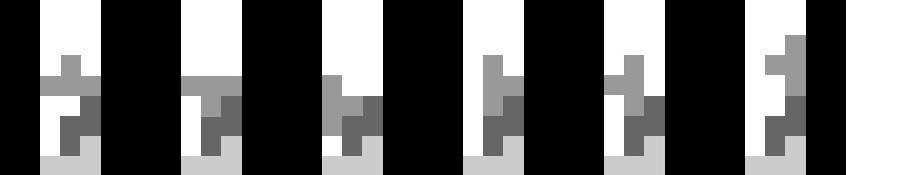}
  \label{fig:T_closed}
\end{figure}

This figure shows the six different ways a T can be placed inside a closed bucket

If we look at the four way to put the T inside the bucket (the fourth one starting from the left), we see that the piece doesnt create gaps that cannot be filled, but it creates gaps on the left. The only way to clear that space without creating gaps that cannot be filled later on, is with a "I" piece.
But as we showed in lemma two, if we try to cover gaps with the I piece, we will not be able to clear the board

In the rest, we will create gaps that cannot be filled. And as it was shown in lemma four, we will not be able to clear a board with gaps that cannot be filled

\end{proof}

\begin{lema}
\label{seven}
If we put a "left snake" piece inside a closed bucket we will not be able to clear the board
\end{lema}

\begin{proof}

We will show this with a figure:

\begin{figure}[H]
  \centering
    \includegraphics[scale=0.5]{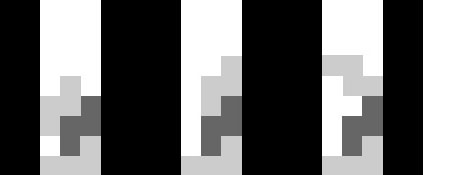}
  \label{fig:snake_closed}
\end{figure}

This figure shows the three different ways a left snake can be placed inside a closed bucket

If we look at the second way to put the snake inside the bucket, we see that the piece doesnt create gaps that cannot be filled, but it creates gaps on the left. The only way to clear that space without creating gaps that cannot be filled later on, is with a "I" piece.
But as we showed in lemma two, if we try to cover gaps with the I piece, we will not be able to clear the board

In the rest of ways to put the piece, we will create  gaps that cannot be filled. And as it was shown in lemma four, we will not be able to clear a board with gaps that cannot be filled

\end{proof}

\begin{lema}
\label{eight}
If we try to put the pieces of a digit in a way that is different that the one we specified in the previous section, we will create  gaps that cannot be filled
\end{lema}

\begin{proof}

We will show this through a series of figures, each depicting all the ways to put a piece of the digit:

\begin{figure}[H]
  \centering
    \includegraphics[scale=0.5]{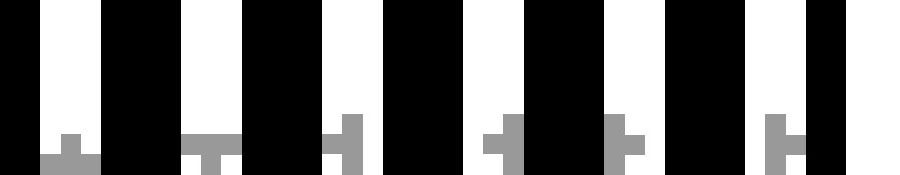}
  \label{fig:first_piece_digit_1}
\end{figure}

This figure shows the six different ways in which the first piece can be placed, as can be seen in the figure, only the first one doesnt create gaps.
And as was shown in lemma six, we cannot put this piece in a closed bucket, without creating gaps.

Now, we are going to see the different ways of placing the second piece of the digit, we will do this with another figure:

\begin{figure}[H]
  \centering
    \includegraphics[scale=0.5]{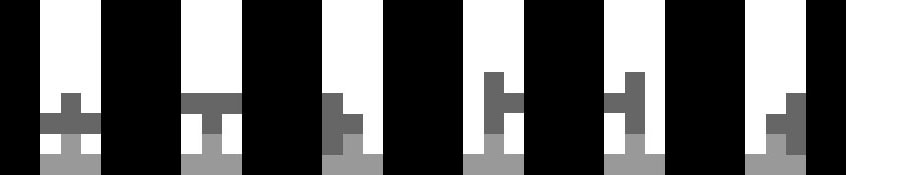}
  \label{fig:second_piece_digit_1}
\end{figure}

This figure shows the six different ways in which the second T that belongs to the digit can be placed in the bucket, as is shown in the figure, only the third from the left and the sixth from the left doesnt create gaps.
But as we will see later, only in the third way, can the right snake by placed without creating gaps.
And as was shown in lemma six, we cannot put this piece in a closed bucket, without creating gaps.

Finally we are going to show that the right gun can only be put inside the opened bucket only if we have placed the second piece in the correct way.

We will do this with two figures, one for every way to put the second piece:

\begin{figure}[H]
  \centering
    \includegraphics[scale=0.5]{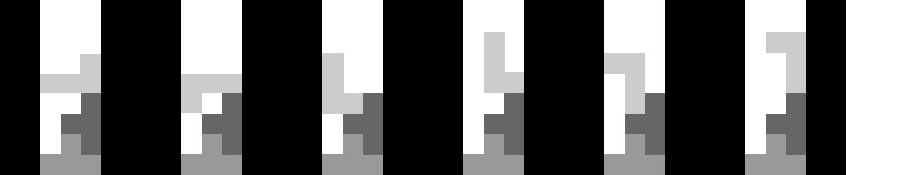}
  \label{fig:third_piece_digit_1}
\end{figure}

This figure shows the first six different ways in which the third piece in the digit can be placed in the bucket. As is shown in the figure,we always create gaps
And as was shown in lemma five, we cannot put this piece in a closed bucket, without creating gaps.

\begin{figure}[H]
  \centering
    \includegraphics[scale=0.5]{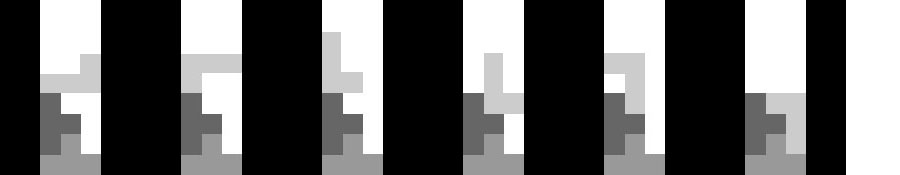}
  \label{fig:fourth_piece_digit_1}
\end{figure}

This figure shows the last six different ways in which the third piece in the digit can be placed in the bucket. As is shown in the figure, the only way to put the piece in the bucket without creating gaps is the sixth one from the left. An as can be seen, it coincides with the way to place the pieces of a digit that we showed in the third section
And as was shown in lemma five, we cannot put this piece in a closed bucket without creating gaps.

\end{proof}

\begin{lema}
\label{nine}
If we try to close the bucket in a way that is different that the one we specified in the previous section, we will create  gaps that cannot be filled
\end{lema}

\begin{proof}

We will show this through a series of figures, each depicting all the ways to put a piece of the closing sequence:

\begin{figure}[H]
  \centering
    \includegraphics[scale=0.5]{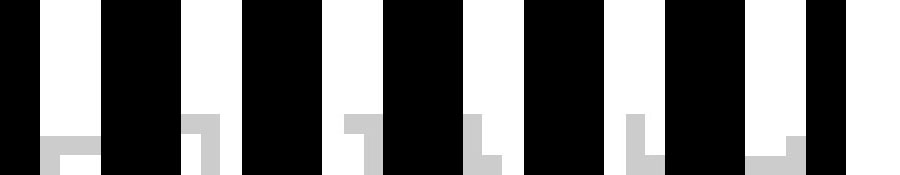}
  \label{fig:first_piece_ending}
\end{figure}

As shown in the figure, the fourth, fifth and sixth way to put the first piece are valid.
And as was shown in lemma five, we cannot put this piece in a closed bucket, without creating gaps.

Now, we are going to show that if we put the left snake in top of the fourth and fifth "right guns" we will create gaps that cannot be filled.

We will do this through a series of figures.

\begin{figure}[H]
  \centering
    \includegraphics[scale=0.5]{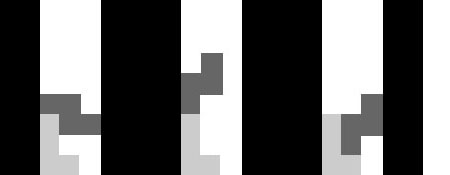}
  \label{fig:second_piece_ending}
\end{figure}

\begin{figure}[H]
  \centering
    \includegraphics[scale=0.5]{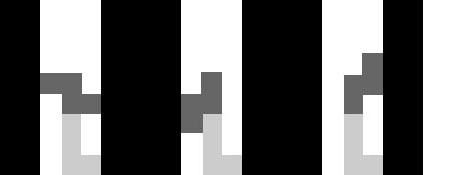}
  \label{fig:third_piece_ending}
\end{figure}

\begin{figure}[H]
  \centering
    \includegraphics[scale=0.5]{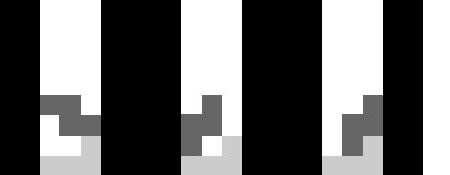}
  \label{fig:fourth_piece_ending}
\end{figure}

As can be seen on the figures above, only the last way to place the right gun and the left snake are the only way to close the bucket without creating gaps.The same that we shown in section three

And as was shown in lemma seven, we cannot put the left snake in a closed bucket, without creating gaps.

\end{proof}

\begin{lema}
\label{ten}
There is only one way to open a bucket without creating gaps
\end{lema}

\begin{proof}
In the following figure we show the six different ways of placing a left gun in a closed bucket.
Only one opens the bucket, the one that doesnt create gaps (the sixth one from the left)

\begin{figure}[H]
  \centering
    \includegraphics[scale=0.5]{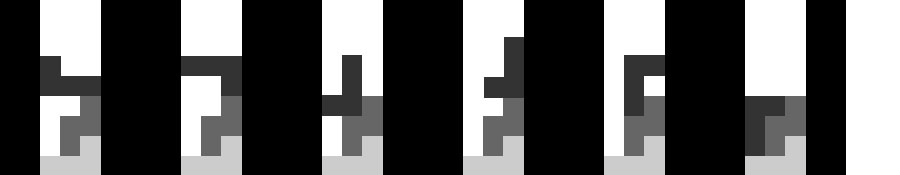}
  \label{fig:open_bucket_piece}
\end{figure}

\end{proof}

\begin{teorema}
\label{ending}
Clearing the board in this Tetris variant is NP-hard
\end{teorema}

\section{Acknowledgements}

We wish to thank Michael Wehar and the community of computer science stack exchange (cs.stackexchange.com) for helpful criticism and commentary

\end{document}